\documentclass[journal,11pt,onecolumn]{IEEEtran}
\usepackage[utf8]{inputenc}

\author{Linan Huang and Quanyan Zhu \\
\ \\
 Department of Electrical and Computer Engineering, Tandon School of Engineering, \\
New York University, Brooklyn, NY 11201 USA \\
Email: \{lh2328, qz494\}@nyu.edu
}

\usepackage{geometry}
\usepackage{graphicx}
\usepackage{enumerate}
\usepackage{graphicx}
\usepackage{amsfonts}
\usepackage{amsmath}
\usepackage{amssymb}
\usepackage{amsthm}
\usepackage{caption,subcaption}
\usepackage{xcolor}

\usepackage[ruled,vlined,linesnumbered,noresetcount]{algorithm2e}

\newtheorem{theorem}{Theorem}

\newtheorem{lemma}{Lemma}

\newtheorem{proposition}{Proposition}
\newtheorem{definition}{Definition}

\begin{document}

\title{Convergence of Bayesian Nash Equilibrium in Infinite Bayesian Games under Discretization}

\maketitle

\begin{abstract}
  We prove the existence of Bayesian Nash Equilibrium (BNE) of general-sum Bayesian games with continuous types and finite actions under the conditions that the utility functions and the prior type distributions are continuous concerning the players' types. 
  Moreover, there exists a sequence of discretized Bayesian games whose BNE strategies converge weakly to a BNE strategy of the infinite Bayesian game.  
  Our proof establishes a connection between the equilibria of the infinite Bayesian game and those of finite approximations, which leads to an algorithm to construct $\varepsilon$-BNE of infinite Bayesian games by discretizing players' type spaces. 
  
\end{abstract}

\section{Introduction}
Bayesian games \cite{harsanyi1967games} have found wide application in auctions \cite{krishna2009auction},  wireless networks \cite{akkarajitsakul2011distributed}, cybersecurity \cite{huang2018analysis,huang2019adaptive}, and robotic systems \cite{huang2019dynamic}. 
In these applications, it is natural to model the incomplete information such as players' bids in auction theory as a continuous random variable. 
However, the existing computational techniques are mainly for finite Bayesian games where the action and the type spaces are both finite. 
For Bayesian games with continuous types, the equilibrium is usually computed under restrictive assumptions. 
For example, \cite{athey2001single} focuses on the single crossing condition and   
the authors in \cite{dd} restrict the type distribution to be piecewise linear with some prior domain knowledge of a qualitative model. 
Iterative methods and learning have also been applied. 
The authors in \cite{reeves2004computing} focus on the piecewise uniform type distribution and payoffs that are linear functions from players' types and actions. They apply an iterated best response to compute the BNE. 
The authors in \cite{rabinovich2013computing} restrict each player's utility to be independent of other's types and develop a fictitious play algorithm to learn pure-strategy equilibrium.

In this paper, we consider general Bayesian games with continuous types and prove the existence of BNE in these games. 
Comparing to previous works (see e.g., \cite{milgrom1985distributional,carbonell2018existence}) that prove the  existence of BNE in infinite Bayesian games, we further prove that there exists a sequence of discretized Bayesian games whose BNE strategies converge weakly to a BNE strategy of the infinite Bayesian game.  
Our  proof  further implies an algorithm to approximate the BNE of infinite Bayesian games by discretization. 
The convergence of equilibrium strategies by discretization or sampling has been shown in complete information games with continuous actions \cite{owen1976existence}, signaling games of certain classes \cite{manelli1996convergence}, and infinite Bayesian Stackelberg games \cite{kiekintveld2011approximation}. 
The authors in \cite{armantier2008approximation} define a new concept of constrained strategic equilibrium (CSE) for Bayesian games and propose sufficient conditions under which a sequence of CSEs converges toward a BNE. 
However, the convergence of BNE has not been shown in simultaneous-move Bayesian games of continuous types. 

After a proper reformulation, we obtain BNE in its distributional form, which enables us to adopt the key idea from \cite{owen1976existence}. Following a similar argument in \cite{owen1976existence}, our results in two-player general-sum infinite Bayesian games can be directly extended to the $N$-player case. 
Since there exists a one-to-one mapping from any set with the cardinality of the continuum to the unit interval $[0,1]$ (i.e., $\mathbb{R}^n$ and $[0,1]$ has the same cardinality), we can directly extend the convergence theorem to any compact joint type space of higher dimensions. 

\section{Bayesian Games with Continuous Types}
We consider the following Bayesian game $\Gamma:=<\mathcal{X},\mathcal{Y}, \Theta_1\times \Theta_2,b(\cdot), \allowbreak
\{\bar{u}^{x,y}(\cdot),\bar{v}^{x,y}(\cdot)\}_{x\in\mathcal{X},y\in \mathcal{Y}}>$ with a compact joint\footnote{The joint type space refers to the Cartesian product (denoted as $\times$) of each player $i$'s type space $\Theta_i$. Since the joint type space is compact, each $\Theta_i$ has to be compact.
}
type space $\Theta_1\times \Theta_2$ and two finite action spaces of $\mathcal{X}:=\{x_1,...,x_L\}$ and $\mathcal{Y}:=\{y_1,...,y_H\}$; i.e., the first and the second player have $L$ and $H$ actions to choose from and simultaneously take action $x\in \mathcal{X}$ and $y\in \mathcal{Y}$, respectively.   
The incomplete information of the game is represented by two single-dimensional continuous random variables $\tilde{\theta}_1\in \Theta_1, \tilde{\theta}_2\in \Theta_2$ whose joint distribution $b$ is assumed to be common knowledge and continuous over the joint type space $\Theta_1\times\Theta_2$. 
We require the marginal distribution to be positive, i.e.,  $\bar{b}_i(\theta_i):=\int_{\Theta_j} b(\theta_i,\theta_j)d\theta_j>0, \forall i\in \{1,2\}, \forall \theta_i\in \Theta_i$ and take $\Theta_1=\Theta_2=[0,1]$ without loss of generality. 
Player $i$ privately observes his type realization $\theta_i\in \Theta_i$ and knows that the other player $j$ has a type $\theta_j\in \Theta_j$ with a probability density of $b_i(\theta_j|\theta_i):= b(\theta_j,\theta_i)/ \bar{b}_i(\theta_i) \in \mathbb{R}^+_0$. 
Then, $b_i$ is a valid  conditional probability measure and we have $\int_{0}^1 b_i(\theta_j|\theta_i)d\theta_j=1, \forall \theta_i\in \Theta_i$. 

The utility functions $\bar{u}^{x,y}(\theta_1,\theta_2)\in \mathbb{R}^+_0$ and $\bar{v}^{x,y}(\theta_1,\theta_2)\in \mathbb{R}^+_0$ of the first and the second player, respectively, depend on players' actions $x\in \mathcal{X},y\in \mathcal{Y}$, and types $\theta_1\in \Theta_1,\theta_2\in \Theta_2$. 
We further assume that both players' utility functions  $\bar{u}^{x,y}(\theta_1,\theta_2)$ and $\bar{v}^{x,y}(\theta_1,\theta_2)$ are continuous over the joint type set $\Theta_1 \times \Theta_2$ for all actions $x\in \mathcal{X},y\in \mathcal{Y}$. 
Since \textit{a continuous function on a compact metric space is bounded and uniformly continuous}, we know that both players' utility functions  are bounded and uniformly continuous over the joint type set.  
Therefore, we can assume non-negative utility functions without loss of generality as we can always add a sufficiently large constant, which is guaranteed by the boundedness, to make them non-negative without any change to the equilibrium policy. 

The  behavioral strategies $\sigma_1:\Theta_1\mapsto \Delta \mathcal{X}$ and $\sigma_2:\Theta_2\mapsto  \Delta \mathcal{Y}$ of the first and the second player, respectively, map each player's type to the distribution of his action space. 
In particular, we denote $\sigma_1(x|\theta_1)\in \mathbb{R}^+_0$ (resp.  $\sigma_2(y|\theta_2)\in \mathbb{R}^+_0$) as the probability of player $1$ (resp. player $2$) taking action $x\in \mathcal{X}$ (resp. action $y\in \mathcal{Y}$) when his type is $\theta_1\in\Theta_1$ (resp. $\theta_2\in\Theta_2$). 
Obviously, we have $\sum_{x\in \mathcal{X}} \sigma_1(x|\theta_1)=1, \forall \theta_1\in \Theta_1$ and  $\sum_{y\in \mathcal{Y}} \sigma_2(y|\theta_2)=1, \forall \theta_2\in \Theta_2$. 
Define two players' expected utilities under any strategy pair $(\sigma_1,\sigma_2)$ as
\begin{equation}
\label{eq:r1r2}
     \begin{split}
         r_1(\theta_1,\sigma_1,\sigma_2):=\int_0^1 b_1(\theta_2|\theta_1) \sum_{x\in\mathcal{X}} \sigma_1 (x|\theta_1)  \sum_{y\in\mathcal{Y}} \sigma_2 (y|\theta_2) \bar{u}^{x,y}(\theta_1,\theta_2)  d\theta_2. \\
          r_2(\theta_2,\sigma_1,\sigma_2):=\int_0^1 b_2(\theta_1|\theta_2) \sum_{x\in\mathcal{X}} \sigma_1^* (x|\theta_1)  \sum_{y\in\mathcal{Y}} \sigma_2 (y|\theta_2) \bar{v}^{x,y}(\theta_1,\theta_2)  d\theta_1. 
     \end{split}
\end{equation}
  For player $i$ of type $\theta_i \in \Theta_i$, his
  best response strategy  $\sigma^*_i(\cdot|\theta_i)$  with respect to the  player's strategy $\sigma_j$ belongs to a set $\mathcal{B}_i(\theta_i,\sigma_j)$, i.e., 
\begin{equation}
\begin{split}
\label{eq:bestresponse1}
\sigma_i^*(\cdot|\theta_i)\in \mathcal{B}_i(\theta_i,\sigma_j):= \arg \max_{\sigma_i(\cdot|\theta_i)}  r_i(\theta_i,\sigma_i,\sigma_j). 
\end{split}
\end{equation}
For any given policy $\sigma_j$ of the other player $j$, player $i$'s best response set $\mathcal{B}_i(\theta_i,\sigma_j)$ under type $\theta_i$ is nonempty and contains a pure policy as shown in Lemma \ref{lemma:purepolicy}. Analogous  statement holds for player $2$.  
\begin{lemma}[Pure Policy in Best Response Set]
\label{lemma:purepolicy}
If the second player's strategy $\sigma_2$ is common knowledge, then player $1$'s best response set $\mathcal{B}_1(\theta_1,\sigma_2)$ under any $\theta_1\in \Theta_1$ contains the following pure policy 
\begin{equation*}
    arg\max_{x\in \mathcal{X}}  \int_0^1   \int_0^1 b_1(\theta_2|\theta_1)  \sum_{y\in\mathcal{Y}} \sigma^*_2 (y|\theta_2) \bar{u}^{x,y}(\theta_1,\theta_2)  d\theta_2.  
\end{equation*} 
\end{lemma}

A strategy pair consists a BNE if they are best response to each other as defined below. 
\begin{definition}[Bayesian Nash Equilibrium]
\label{def:BNE}
    A strategy pair $(\sigma_1^*,\sigma_2^*)$ consists a BNE of infinite Bayesian game $\Gamma$ if
    $\sigma_i^*(\cdot|\theta_i)\in \mathcal{B}_i(\theta_i,\sigma^*_j), \forall i,j\in \{1,2\}, i\neq j$, for almost\footnote{“Almost” in this context means that the probability of all types for which the strategy does not prescribe an optimal action is zero. For example, if player $i$'s strategies differ only at countable points over $\Theta_i$, then they result in the same value of Riemann integration in \eqref{eq:r1r2}. } every $\theta_1\in \Theta_1$ and  $\theta_2\in \Theta_2$. 
    \end{definition} 
Since $\bar{b}_i(\theta_i)>0, \forall \theta_i\in \Theta, \forall i\in \{1,2\}$, Lemma \ref{lemma:BNE=NE} below shows that we can compute BNE strategy pair $(\sigma_1^*,\sigma_2^*)$ through the following integration form in \eqref{eq:bestresponse1integral} and \eqref{eq:bestresponse2integral}; i.e., no player has a profitable deviation after he knows his private type if and only if he does not benefit from any deviation before knowing his type \cite{harsanyi1967games}.  
\begin{equation}
\begin{split}
\label{eq:bestresponse1integral}
 \int_0^1 \int_0^1 b (\theta_1,\theta_2)    \sum_{x\in\mathcal{X}} \sigma^*_1(x|\theta_1)  \sum_{y\in\mathcal{Y}} \sigma^*_2 (y|\theta_2) \bar{u}^{x,y}(\theta_1,\theta_2)  d\theta_1  d\theta_2 \\
      = \max_{\sigma_1}
     \int_0^1   \int_0^1 b (\theta_1,\theta_2) \sum_{x\in\mathcal{X}} \sigma_1 (x|\theta_1)  \sum_{y\in\mathcal{Y}} \sigma^*_2 (y|\theta_2) \bar{u}^{x,y}(\theta_1,\theta_2)  d\theta_1   d\theta_2     ,
\end{split}
\end{equation}
and 
\begin{equation}
\begin{split}
\label{eq:bestresponse2integral}
      \int_0^1  \int_0^1 b (\theta_1,\theta_2)  \sum_{x\in\mathcal{X}} \sigma^*_1(x|\theta_1)  \sum_{y\in\mathcal{Y}} \sigma^*_2 (y|\theta_2) \bar{v}^{x,y}(\theta_1,\theta_2)  d\theta_1 d\theta_2 \\
      = \max_{\sigma_2}
      \int_0^1 \int_0^1 b (\theta_1,\theta_2)  \sum_{x\in\mathcal{X}} \sigma_1^* (x|\theta_1)  \sum_{y\in\mathcal{Y}} \sigma_2 (y|\theta_2) \bar{v}^{x,y}(\theta_1,\theta_2)  d\theta_1  d\theta_2    .   
\end{split}
\end{equation}

    \begin{lemma}[BNE is equivalent to Nash Equilibrium]
    \label{lemma:BNE=NE}
     A strategy pair $(\sigma_1^*,\sigma_2^*)$ consists a BNE if and only if \eqref{eq:bestresponse1integral} and \eqref{eq:bestresponse2integral} holds. 
    \end{lemma}
     \begin{proof}
     The `only if' part (sufficiency) is straight forward as \eqref{eq:bestresponse1} results in 
     \eqref{eq:bestresponse1integral} and \eqref{eq:bestresponse2integral}. 
     To prove the `if' part (necessity), we show that if $(\sigma_1^*,\sigma_2^*)$ is not a BNE defined in Definition \ref{def:BNE}, then  \eqref{eq:bestresponse1integral} and \eqref{eq:bestresponse2integral} cannot hold in the same time. 
     As $(\sigma_1^*,\sigma_2^*)$ is not a BNE, there exists a measurable set $\hat{\Theta}_i\subseteq \Theta_i$ and  at least one player $i$ (assume the second player) who has a profitable deviation\footnote{Since the best response of any give policy contains a pure policy as shown in Lemma \ref{lemma:purepolicy}, we can restrict the profitable deviation to an action without loss of generality. } 
     from $\sigma_2^*(\cdot|\theta_2)$ to an action $y_{{l}}\in\mathcal{Y}$ when $\theta_2\in \hat{\Theta}_2$, i.e.,
     \begin{equation*}
        \begin{split}
         &\int_{\hat{\Theta}}   \bar{b}_2(\theta_1) 
         \bigg[ \int_0^1 b_1(\theta_2|\theta_1) \sum_{x\in\mathcal{X}} \sigma^*_1(x|\theta_1)  \bar{v}^{x,y_{{l}}}(\theta_1,\theta_2)  d\theta_2 \bigg] d\theta_1 \\
        >&
        \int_{\hat{\Theta}}  \bar{b}_2(\theta_1) 
        \bigg[ \int_0^1 b_1(\theta_2|\theta_1)
      \sum_{x\in\mathcal{X}} \sigma^*_1(x|\theta_1)  \sum_{y\in\mathcal{Y}} \sigma^*_2 (y|\theta_2) \bar{v}^{x,y}(\theta_1,\theta_2)  d\theta_2 \bigg] d\theta_1. 
        \end{split}
     \end{equation*}
     Consider a strategy $\hat{\sigma_2}$ where
      $\hat{\sigma_2}(y|\theta_2)=\mathbf{1}_{\{y=y_{{l}}\}}, \forall y\in \mathcal{Y}, \theta_2\in \hat{\Theta}_2$ and $\hat{\sigma_2}={\sigma_2},\forall \theta_2\notin \hat{\Theta}_2$; i.e.,  $\hat{\sigma_2}$ is identical to  ${\sigma_2}$ except over the set $\hat{\Theta}_2$. Then, we know that 
       \begin{equation*}
        \begin{split}
        & \int_0^1 \int_0^1 b(\theta_1,\theta_2) \sum_{x\in\mathcal{X}} \sigma^*_1(x|\theta_1)  \sum_{y\in\mathcal{Y}} \hat{\sigma_2} (y|\theta_2) \bar{v}^{x,y}(\theta_1,\theta_2)  d\theta_1 d\theta_2 \\
        =& \int_{\hat{\Theta_2}} \hat{b}_2(\theta_1)  \bigg[\int_0^1 b_1(\theta_2|\theta_1) \sum_{x\in\mathcal{X}} \sigma^*_1(x|\theta_1)  \sum_{y\in\mathcal{Y}} \hat{\sigma_2} (y|\theta_2) \bar{v}^{x,y}(\theta_1,\theta_2)  d\theta_2  \bigg] d\theta_1 \\
        +& \int_{\Theta_2\setminus \hat{\Theta}_2} \hat{b}_2(\theta_1) \bigg[ \int_0^1 b_1(\theta_2|\theta_1) \sum_{x\in\mathcal{X}} \sigma^*_1(x|\theta_1)  \sum_{y\in\mathcal{Y}} \sigma^*_2 (y|\theta_2) \bar{v}^{x,y}(\theta_1,\theta_2)  d\theta_2 \bigg] d\theta_1 \\
        >&    \int_0^1  \int_0^1  b(\theta_1,\theta_2) \sum_{x\in\mathcal{X}} \sigma^*_1(x|\theta_1)  \sum_{y\in\mathcal{Y}} {\sigma^*_2} (y|\theta_2) \bar{v}^{x,y}(\theta_1,\theta_2)  d\theta_1 d\theta_2, 
        \end{split}
     \end{equation*}    
     which contradicts \eqref{eq:bestresponse2integral}. 
     \end{proof}

\subsection{Equivalent Reformulation in Distributional Form}
\label{sec:reformulation}
Since both integrands in \eqref{eq:bestresponse1integral} and \eqref{eq:bestresponse2integral} are non-negative, we can exchange the summation of actions and the integration of types according to Fubini's theorem. 
Define ${u}^{x,y}(\theta_1,\theta_2):=b(\theta_1,\theta_2) \bar{u}^{x,y}(\theta_1,\theta_2)$ and ${v}^{x,y}(\theta_1,\theta_2):=b(\theta_1,\theta_2) \bar{v}^{x,y}(\theta_1,\theta_2)$. 
Since a finite production of continuous functions is still continuous, ${u}^{x,y}$ and ${v}^{x,y}$ are both continuous over the joint set $\Theta_1\times\Theta_2$.  
By assimilating the prior distribution of types into the players' utility functions, we can discretize the continuous type set uniformly as shown Section \ref{sec:owen}.
Let represent the first player's behavioral strategy $\sigma_1(x|\theta_1)$ as a 
function of $\theta_1$ parameterized by action $x$, i.e., $f^x(\theta_1)$.  
Then we can define a non-decreasing bounded function $F^x(\theta_1):=\int_{0}^{{\theta}_1} f^x(\tilde{\theta}_1)d\tilde{\theta}_1 $ of $\theta_1$ parameterized by action $x$. 
Since $\sum_{x\in \mathcal{X}} f^x(\theta_1)=1, \forall \theta_1\in \Theta_1$, and $f^x(\theta_1)\geq 0, \forall \theta_1, \forall x\in \mathcal{X}$, we obtain $\sum_{x\in \mathcal{X}} F^x(\theta_1)=\theta_1, \forall \theta_1\in \Theta_1$ by Fubini's theorem. We use $\mathcal{F}^{\mathcal{X}}$ to denote the set of functions $F^{\mathcal{X}}:=\{F^{x}\}_{x\in \mathcal{X}}$ that satisfy the above conditions. 
Similarly, we can represent the second player's strategy $\sigma_2(y|\theta_2)$ as $g^y(\theta_2)$ and define $G^y(\theta_2):=\int_{0}^{{\theta}_2} f^x(\tilde{\theta}_2)d\tilde{\theta}_2$ as the non-decreasing bounded function of $\theta_2$. 
Analogously, we have $G^y(0)=0$ for any $y\in \mathcal{Y}$ and $\sum_{y\in \mathcal{Y}} G^y(\theta_2)=\theta_2, \forall \theta_2\in \Theta_2$. 
We use $\mathcal{G}^{\mathcal{Y}}$ to denote the set of functions $G^{ \mathcal{Y}}:=\{G^y\}_{y\in \mathcal{Y}}$ that satisfy the above conditions. 
Then, we can recast a BNE strategy pair $(F_0^{ \mathcal{X}}\in \mathcal{F}^{\mathcal{X}},G_0^{ \mathcal{Y}}\in \mathcal{G}^{\mathcal{Y}})$ in the following \textit{distributional form}, i.e.,  
\begin{equation*}
\begin{split}
\label{eq:BNEinVector}
     & \sum_{x\in\mathcal{X}} \sum_{y\in\mathcal{Y}}     \int_0^1   \int_0^1  {u}^{x,y}(\theta_1,\theta_2)  dF_0^x(\theta_1) dG_0^y(\theta_2)
     = \max_{F^{\mathcal{X}} \in \mathcal{F}^{\mathcal{X}} } \sum_{x\in\mathcal{X}} \sum_{y\in\mathcal{Y}}     \int_0^1   \int_0^1  {u}^{x,y}(\theta_1,\theta_2)  dF^x(\theta_1) dG_0^y(\theta_2), 
     \\
     &
     \sum_{x\in\mathcal{X}} \sum_{y\in\mathcal{Y}}     \int_0^1   \int_0^1  {v}^{x,y}(\theta_1,\theta_2)  dF_0^x(\theta_1) d G_0^y(\theta_2)
     = \max_{G^{\mathcal{Y}} \in \mathcal{G}^{\mathcal{Y}} } \sum_{x\in\mathcal{X}} \sum_{y\in\mathcal{Y}}     \int_0^1   \int_0^1  {v}^{x,y}(\theta_1,\theta_2)  d F_0^x(\theta_1) dG^y(\theta_2).  
     \end{split}
\end{equation*}

Due to the difficulty of computing an exact BNE, it is common to consider an approximate equilibrium defined below. 
\begin{definition}[$\varepsilon$-BNE]
\label{def:epsBNE}
A strategy pair $(F_0^{ \mathcal{X}}\in \mathcal{F}^{\mathcal{X}},G_0^{ \mathcal{Y}}\in \mathcal{G}^{\mathcal{Y}})$ consists a $\varepsilon$-BNE if for all $(F^{ \mathcal{X}}\in \mathcal{F}^{\mathcal{X}},G^{ \mathcal{Y}}\in \mathcal{G}^{\mathcal{Y}})$, the following holds. 
\begin{equation*}
\begin{split}
\label{eq:eps_BNE}
     & \sum_{x\in\mathcal{X}} \sum_{y\in\mathcal{Y}}     \int_0^1   \int_0^1  {u}^{x,y}(\theta_1,\theta_2)  dF_0^x(\theta_1) dG_0^y(\theta_2)
     \geq \sum_{x\in\mathcal{X}} \sum_{y\in\mathcal{Y}}     \int_0^1   \int_0^1  {u}^{x,y}(\theta_1,\theta_2)  dF^x(\theta_1) dG_0^y(\theta_2)- \varepsilon, 
     \\
     &
     \sum_{x\in\mathcal{X}} \sum_{y\in\mathcal{Y}}     \int_0^1   \int_0^1  {v}^{x,y}(\theta_1,\theta_2)  dF_0^x(\theta_1) d G_0^y(\theta_2)
     \geq  \sum_{x\in\mathcal{X}} \sum_{y\in\mathcal{Y}}     \int_0^1   \int_0^1  {v}^{x,y}(\theta_1,\theta_2)  d F_0^x(\theta_1) dG^y(\theta_2) - \varepsilon.  
     \end{split}
\end{equation*}
\end{definition}

\section{Extension of Helly's Selection Theorem}
Based on the reformulation of BNE in Section \ref{sec:reformulation}, the players' strategies $F^{\mathcal{X}}$ and $G^{\mathcal{Y}}$ become $L$- and $H$-dimensional vectors of constrained functions, respectively. 
Thus, we extend the original Helly's selection theorem in the following lemma to fit the vector of functions with constraints. 


\begin{lemma}[Convergence on Countable Set]
\label{lemma:countable}
Consider the finite  set $\mathcal{X}:=\{x_1,...,x_L\}$ and a sequence of functions $\{F^{\mathcal{X}}_n\}_{n\in \mathcal{Z}^+}$, where $F_n^{\mathcal{X}}\in \mathcal{F}^{\mathcal{X}}$ for each $n\in \mathcal{Z}^+$. Let $\mathcal{D}:=\{\theta^1,\theta^2,...\}$ be any countable subset of $ \Theta$. Then there is a subsequence of  $\{F^{\mathcal{X}}_n\}_{n\in \mathcal{Z}^+}$, i.e., $\{F^{\mathcal{X}}_{n_k}\}_{k=1}^{\infty}$ such that $\bar{F}_0^{x}(\theta):= \lim_{k\rightarrow \infty} F^{x}_{n_k}(\theta)$ exists for any $\theta\in \mathcal{D},  x\in \mathcal{X}$. Moreover, the limit function $\bar{F}_0^{\mathcal{X}}\in \mathcal{F}^{\mathcal{X}}$. 
\end{lemma}

\begin{proof}
With a little abuse of notation, the vector $F^{\mathcal{X}}_{n}(\theta^d):=[F^{x_1}_n(\theta^d),\cdots,  F^{x_L}_n(\theta^d)]$ at $\theta^d\in \mathcal{D}, d\in \{1,2,\cdots\}, \forall n\in \mathbb{Z}^+$, belongs to a subset of $\mathbb{R}^L$, i.e., $\mathcal{F}^{\mathcal{X}}(\theta^d)$, 
that is closed and bounded. 
Then the subset must be  sequentially compact based on Bolzano–Weierstrass theorem and every sequence of points in this subset has a convergent subsequence to a point in the subset. 
Thus, we know that there exist a subsequence $F^{\mathcal{X}}_{n^d_k}(\theta^d)$ converge to $\bar{F}_0^{\mathcal{X}}(\theta^d)\in \mathcal{F}^{\mathcal{X}}(\theta^d)$.  
Then, we can apply the standard diagonalization argument to repeatedly find subsequence from subsequence so that there exist a final subsequence $n_k$ that makes $F^{\mathcal{X}}_{n_k}(\theta)$ converges to $\bar{F}_0^{\mathcal{X}}(\theta)\in \mathcal{F}^{\mathcal{X}}(\theta)$ for all $\theta\in \mathcal{D}$.  
Note that $\bar{F}_0^{x}$ is non-decreasing with respect to $\theta\in \mathcal{D}$ for each $x\in \mathcal{X}$ as the inequality is preserved in the limit; i.e., if $\theta^{d_1}<\theta^{d_2}\in \mathcal{D}$, then $\bar{F}_0^{x}(\theta^{d_1})=\lim_{k\rightarrow \infty} F^{x}_{n_k}(\theta^{d_1})\leq \lim_{k\rightarrow \infty} F^{x}_{n_k}(\theta^{d_2})=\bar{F}_0^{x}(\theta^{d_2})$.  
\end{proof}

\begin{theorem}[Convergence on Compact Set]
\label{thm:first thm}
Consider  $\Theta:=[0,1]$,  finite  set $\mathcal{X}:=\{x_1,...,x_L\}$, and a sequence of functions $\{F^{\mathcal{X}}_n\}_{n\in \mathcal{Z}^+}$, where $F_n^{\mathcal{X}}\in \mathcal{F}^{\mathcal{X}}$ for each $n\in \mathcal{Z}^+$. 
Then, some subsequence of $\{F^{\mathcal{X}}_n\}_{n\in \mathcal{Z}^+}$, i.e., $\{F^{\mathcal{X}}_{n_k}\}_{k=1}^{\infty}$, converges point-wise to a non-decreasing bounded function $F^{\mathcal{X}}_0\in \mathcal{F}^{\mathcal{X}}$, i.e., $F^{x}_{0}(\theta)= \lim_{k\rightarrow \infty} F^{x}_{n_k}(\theta), \forall x\in \mathcal{X}, \forall \theta\in \Theta$. 
\end{theorem}

\begin{proof}
Let $\mathcal{D}:=\mathbb{Q} \cap [0,1]$, then $\mathcal{D}$ is countable where $\mathbb{Q}$ represents the set of rational number. 
Then, based on Lemma \ref{lemma:countable}, there exists a subsequence  $\{F^{\mathcal{X}}_{n_{k^1}}\}_{k^1=1}^{\infty}$ that converges to $\bar{F}_0^{\mathcal{X}}\in \mathcal{F}^{\mathcal{X}}$ if $\theta\in \mathcal{D}$. 
Next, we need to extend the function $\bar{F}_0^{\mathcal{X}}$ defined on discrete set $\mathcal{D}$ to a function $F^{\mathcal{X}}_0$ defined over the continuous region by connecting the dots, i.e., $F^{x}_0(\alpha)=\sup_{\beta\leq \alpha, \beta\in \mathcal{D}} \bar{F}_0^x(\beta), \forall \alpha\in \Theta, \forall x\in \mathcal{X}$. 
Then, $F^{\mathcal{X}}_0(\theta)=\bar{F}_0^{\mathcal{X}}(\theta), \forall \theta \in {\mathcal{D}}$, and $F^{\mathcal{X}}_0\in \mathcal{F}^{\mathcal{X}}$ is also element-wise non-decreasing with respect to $\theta\in \Theta$
as $\alpha < \gamma$ leads to 
$
    F^{x}_0(\alpha)=\sup_{\beta\leq \alpha, \beta\in \mathcal{D}} \bar{F}_0^x(\beta)\leq \sup_{\beta\leq \gamma, \beta\in \mathcal{D}} \bar{F}_0^x(\beta)= F^{x}_0(\gamma), \forall x\in \mathcal{X}. 
$

Note that by connecting discrete dots, $F^{\mathcal{X}}_0(\theta)$ is right continuous and there are countable jumps at $\theta\in \mathcal{D}$. We first show that for each $x\in \mathcal{X}$, if $F^{x}_0$ is continuous at $\alpha\in \Theta$, then there exists a subsequence $n_{k^2}$ of the subsequence $n_{k^1}$ such that $F^{x}_0(\alpha)=\lim_{k^2\rightarrow\infty} F^{x}_{n_{k^2}}(\alpha)$. 
For any $\alpha\in \Theta$, since  $F^{x}_0, \forall x\in \mathcal{X}$, is continuous at $\alpha\in \Theta$, we can choose $p,q\in \mathcal{D}$, $\alpha\in (p,q)$  such that $F^{x}_0(q)-F^{x}_0(p)<\epsilon/2, \forall x\in \mathcal{X}$. 
Owning to the convergence on the countable set $\mathcal{D}$, we can pick $k^2$ sufficiently large such that $F^{x}_{n_{k^2}}(p)\in (F^{x}_0(p)-\epsilon/2,F^{x}_0(p)+\epsilon/2)$ and $F^{x}_{n_{k^2}}(q)\in (F^{x}_0(q)-\epsilon/2,F^x_0(q)+\epsilon/2)$ for all $x\in \mathcal{X}$.  
Then, \begin{equation*}
        F^{x}_{n_{k^2}}(\alpha)
        \leq F^{x}_{n_{k^2}}(q)<F^x_0(q)+\epsilon/2
        <F^{x}_0(p)+\epsilon
        \leq F^{x}_0(\alpha)+\epsilon, \forall x\in \mathcal{X}. 
    \end{equation*}
Analogously, we can also obtain     
$F^{x}_{n_{k^2}}(\alpha)> F^{x}_0(\alpha)-\epsilon , \forall x\in \mathcal{X}$, which together show the convergence at $\alpha$. 
Second, we show the convergence at discontinuous point $\beta\in \Theta$. 
Since $F^{\mathcal{X}}_0$ is element-wise non-decreasing, the set of discontinuity is at most countable based on Froda's theorem.
 Thus, Lemma \ref{lemma:countable} guarantees that we can  select a convergent subsequence $n_{k}$ from the subsequence $n_{k^2}$ such that  $F^{x}_0(\beta)=\lim_{k \rightarrow\infty} F^{x}_{n_{k}}(\beta), \forall x\in \mathcal{X}$. 
 Combining the above two cases, we have found a convergent subsequence $n_k$ over the entire set $\Theta$, i.e., $F^{x}_{0}(\theta)= \lim_{k\rightarrow \infty} F^{x}_{n_k}(\theta), \forall x\in \mathcal{X}, \forall \theta\in \Theta$. 
\end{proof}

Next, we extend Helly's second theorem to a production of sets in Theorem \ref{thm:second thm}. 
\begin{theorem}
\label{thm:second thm}
Let $u^{x,y}(\theta_1,\theta_2)$ be continuous over the joint set $[\alpha,\beta]\times \mathcal{C} $ for each action $x\in \mathcal{X},y\in \mathcal{Y}$, where $\mathcal{C}$ is compact and $[\alpha,\beta]\subseteq \Theta$, then for each $y\in \mathcal{Y}$, $\sum_{x\in \mathcal{X}} \int_{\alpha}^{\beta} u^{x,y}(\theta_1,\theta_2)dF^x_n(\theta_1)$ converges to $\sum_{x\in \mathcal{X}}  \int_{\alpha}^{\beta} u^{x,y}(\theta_1,\theta_2)dF^x_0(\theta_1)$ uniformly in $\theta_2$. 
\end{theorem}
\begin{proof}
For any $\epsilon>0$, since  $u^{x,y}$ is uniformly continuous over the joint set, we can choose $\delta>0$ such that 
\begin{equation}
\label{eq:(2)}
    |u^{x,y}(\theta_1,\theta_2)-u^{x,y}(\theta'_1,\theta_2)|<\epsilon
\end{equation}
for all $x\in \mathcal{X},y\in \mathcal{Y}, \theta_2\in \mathcal{C}, |\theta_1-\theta_1'|<\delta$.  
Choose $\alpha=\theta_1^1<\theta_1^2<\cdots<\theta_1^D=\beta$ such that $F^x_n$ is continuous at each $\theta_1^d, d\in \{2,\cdots,D-1\}$ and $\theta_1^{d+1}-\theta_1^d<\delta$ for all $x\in \mathcal{X}$, which can be done as $F^x_n$ has at most countable discontinuities over the set $[\alpha,\beta]$. 
Define $
    u^{x,y}_d(\theta_2):=\min_{\theta_1^d\leq \theta_1\leq \theta_1^{d+1}} u^{x,y}(\theta_1,\theta_2)$
and 
$  S^{x,y}_n(\theta_2):=\sum_{d=1}^{D-1} \int_{\theta_1^d}^{\theta_1^{d+1}} u^{x,y}_d(\theta_2)dF^x_n(\theta_1), \forall n\in \mathbb{Z}^+_0. $
Then, we have
\begin{equation}
\label{eq:(5)}
    \int_{\alpha}^{\beta}u^{x,y}(\theta_1,\theta_2)d F^x_n(\theta_1) \geq  S^{x,y}_n(\theta_2). 
\end{equation}
Now by \eqref{eq:(2)} and the monotonicity of $F_n^x$, we have 
$
     S^{x,y}_n(\theta_2)\geq \int_{\alpha}^{\beta} u^{x,y}(\theta_1,\theta_2)d F^x_n(\theta_1)-\epsilon (F^x_n(\beta)-F^x_n(\alpha)), \forall x\in \mathcal{X},y\in \mathcal{Y}. 
$
Then, using \eqref{eq:(5)} and the fact that $F^x_n(\beta)-F^x_n(\alpha)\leq 1, \forall x\in \mathcal{X}$, 
we have 
\begin{equation}
\label{eq:(7)}
    | S^{x,y}_n(\theta_2)-\int_{\alpha}^{\beta}  u^{x,y}(\theta_1,\theta_2)d F^x_n(\theta_1)|<\epsilon , \forall x\in \mathcal{X},y\in \mathcal{Y}.
\end{equation}
Now, choose $N$ large enough such that, for each $d=1,\cdots,D$, and each $n\geq N$, 
\begin{equation}
\label{eq:(8)}
    |F_n^{x}(\theta_1^d)-F_0^{x}(\theta_1^d) |<\frac{\epsilon}{D}, \forall x\in \mathcal{X}. 
\end{equation}
Then, we obtain 
$
    \int_{\theta_1^d}^{\theta_1^{d+1}} u_d^{x,y}(\theta_2) dF_n^{x}(\theta_1) = u_d^{x,y}(\theta_2) (F_n^{x}(\theta_1^{d+1})-F_n^{x}(\theta_1^d)), \forall x\in \mathcal{X},y\in \mathcal{Y}, \forall n\in \mathbb{Z}_0^+, 
$
which, together with \eqref{eq:(8)}, gives us 
$
    | \int_{\theta_1^d}^{\theta_1^{d+1}} u_d^{x,y}(\theta_2) dF_n^{x}(\theta_1)- \int_{\theta_1^d}^{\theta_1^{d+1}} u_d^{x,y}(\theta_2) dF_0^{x}(\theta_1) | < \frac{2\epsilon}{D} |u_d^{x,y}(\theta_2)|, \allowbreak
    \forall x\in \mathcal{X},y\in \mathcal{Y}.  
$
Since $u^{x,y}$ is continuous over a compact set, there exists a finite upper bound $M$ for $|u^{x,y}_d(\theta_2)|$. 
Therefore, 
\begin{equation}
\label{eq:(11)}
    | S^{x,y}_n(\theta_2)-S^{x,y}_0(\theta_2) | <\sum_{d=1}^{D-1} \frac{2\epsilon}{D} |u_d^{x,y}(\theta_2)| \leq 2\epsilon M,
\end{equation}
Combine \eqref{eq:(7)} and \eqref{eq:(11)}, we have that $\forall  \theta_2\in \mathcal{C}, n\geq N$, 
\begin{equation*}
    | \int_{\alpha}^{\beta}u^{x,y}(\theta_1,\theta_2)d F^x_n(\theta_1) - \int_{\alpha}^{\beta}u^{x,y}(\theta_1,\theta_2)d F^x_0(\theta_1)| < (2M+1) \epsilon , \forall x\in \mathcal{X},y\in \mathcal{Y},   
\end{equation*}
or equivalently, 
\begin{equation*}
   \sum_{x\in \mathcal{X}} | \int_{\alpha}^{\beta}u^{x,y}(\theta_1,\theta_2)d F^x_n(\theta_1) - \int_{\alpha}^{\beta}u^{x,y}(\theta_1,\theta_2)d F^x_0(\theta_1)| < |\mathcal{X}|\cdot (2M+1)   \epsilon,\forall y\in \mathcal{Y}.  
\end{equation*}
Since $\epsilon$ is arbitrary and its coefficient $|\mathcal{X}|\cdot (2M+1) $ is fixed for all $\theta_2\in \mathcal{C}$, the convergence is uniformly in $\theta_2$ for each $y\in \mathcal{Y}$. 
\end{proof}

\section{Discretization and Convergence}
\label{sec:owen}
In this section, we provide a theoretical guarantee to approximate infinite Bayesian games by properly discretizing the type space and solving the resulted finite Bayesian games.  
The convergence of the BNE is guaranteed as long as the maximum distance of intervals under the discretization scheme goes to zero when the number of intervals goes to infinity. 
For simplicity, we adopt the following uniform discretization scheme. 
We can also adopt other deterministic schemes such as dichotomy or stochastic schemes such as sampling. 

For any integer $n\geq 1$ and action pair $x\in \mathcal{X}, y\in \mathcal{Y}$, define the level-$n$ approximation of two players' utility functions $u^{x,y},v^{x,y}$ as two $n\times n$ matrices $[u^{x,y,n}_{i,j}]_{i,j\in \{1,\cdots,n\}}, [v^{x,y,n}_{i,j}]_{i,j\in \{1,\cdots,n\}}$, respectively, where the $(i,j)$ elements are 
\begin{equation}
\label{eq:discretizaUtility}
    u^{x,y,n}_{i,j}=u^{x,y}(\frac{i}{n},\frac{j}{n}),  v^{x,y,n}_{i,j}=v^{x,y}(\frac{i}{n},\frac{j}{n}). 
\end{equation}
Then, the level-$n$ discretized version of the infinite Bayesian game $\Gamma$ is denoted as 
\begin{equation*}
    \Gamma^n = <\mathcal{X},\mathcal{Y},\bar{\Theta}_1^n\times \bar{\Theta}_2^n  ,b^n(\cdot), \{{u}^{x,y,n}_{i,j},{v}^{x,y,n}_{i,j}\}^{i,j\in \{1,\cdots,n\}}_{x\in\mathcal{X},y\in \mathcal{Y}}>, 
\end{equation*}
where the finite type set $\bar{\Theta}_i^n:=\{\frac{1}{n},\cdots,\frac{n}{n} \}$ contains $n$ discrete types of player $i$. 
Since we have assimilated the the prior type distribution $b(\cdot)$ into the players' utility functions $u^{x,y},v^{x,y}$, the prior distribution of the discrete types is $b^n(\frac{i}{n},\frac{j}{n})=\frac{1}{n^2}, \forall i,j\in \{1,\cdots,n\}$. 
Let $s^{\mathcal{X},n}:=(s^{\mathcal{X},n}_1,\cdots,s^{\mathcal{X},n}_n)$ and $t^{\mathcal{Y},n}:=(t^{\mathcal{Y},n}_1,\cdots,t^{\mathcal{Y},n}_n)$ be a BNE of the level-$n$ discretized Bayesian game $\Gamma^n$ where the elements of $s_i^{\mathcal{X},n}:=[s_i^{x_1,n},s_i^{x_2,n},\cdots,s_i^{x_L,n}]$ and $t_j^{\mathcal{Y},n}:=[t_i^{y_1,n},t_i^{y_2,n},\cdots,t_i^{y_H,n}]$ are all non-negative for all $i,j\in \{1,\cdots,n\}$ and each sum up to be $1$, i.e., $\sum_{l=1}^L s_i^{x_l,n}=1$, $\sum_{h=1}^H t_i^{y_h,n}=1, \forall i,j\in \{1,\cdots,n\}$.  
The existence of behavioral strategy pairs $(s^{\mathcal{X},n},t^{\mathcal{Y},n})$ is guaranteed \cite{shoham_leyton-brown_2008} for any finite Bayesian games $\Gamma^n$. 
For any $x\in \mathcal{X}, y\in \mathcal{Y}$ and  $n\in \mathcal{Z}^+$, define the non-decreasing right-continuous step functions
\begin{equation}
\label{eq: FnGn}
    F^x_n(\theta_1)=\frac{1}{n} \sum_{i=1}^{\lfloor n\theta_1 \rfloor} s_i^{x,n},  \  
    G^y_n(\theta_2)=\frac{1}{n} \sum_{j=1}^{\lfloor n\theta_2 \rfloor} t_j^{y,n}, 
\end{equation}
where $\lfloor n\theta_1 \rfloor$ represents the great integer that is not greater than the value of $n\theta_1$. 
Obviously, $F^{\mathcal{X}}_n\in \mathcal{F}^{\mathcal{X}}$ and $G^{\mathcal{Y}}_n\in \mathcal{G}^{\mathcal{Y}}$ for any $n\in \mathcal{Z}^+$. 

Since player $2$ has $H$ possible actions, we can divide the entire type space into at most $H$ disjoint subsets, i.e.,  $\Theta_2=\cup_{h=1}^H \Theta_2^h, \Theta_2^h\cap \Theta_2^{h'}=\emptyset, \forall h\neq h'$, where  player $2$ chooses to take action $y^h\in \mathcal{Y}$ when his type $\theta_2$ belongs to $\Theta_2^h$, i.e., $g^{y_h}(\theta_2)=\mathbf{1}_{\{\theta_2\in \Theta_2^h \}}, \forall h\in \{1,2,...,H\}, \forall \theta_2\in \Theta_2$. 
Note that each subset $\Theta_2^h\subseteq \Theta_2,h\in \{1,\cdots,H\}$, does not need to be connected and can be empty. 

\begin{lemma}
\label{lemma:continuous}
The function $\sum_{h\in \{1,\cdots,H\}} \int_{\Theta_2^h} {v}^{x,y_h}(\theta_1,\theta_2)  d\theta_2 $ is continuous over $\theta_1$ for any $x\in \mathcal{X}$. 
\end{lemma}
\begin{proof}
Since ${v}^{x,y_h}(\theta_1,\theta_2)$ is continuous over the joint type space for any $x\in \mathcal{X}, y\in \mathcal{Y}$, we know that for any number $\epsilon > 0$, however small, there exists some number $\delta > 0$ such that for all $\theta_1\in (\alpha-\delta,\alpha+\delta)$,  ${v}^{x,y_h}(\theta_1,\theta_2) \in ({v}^{x,y_h}(\alpha,\theta_2)-\epsilon, {v}^{x,y_h}(\alpha,\theta_2)+\epsilon )$ for all $\theta_2\in \Theta_2$. 
Based on the fact that $ \sum_{h\in \{1,\cdots,H\}} \int_{\Theta_2^h}  d\theta_2 \equiv 1$, we have
$
    \sum_{h\in \{1,\cdots,H\}} \int_{\Theta_2^h}  {v}^{x,y_h}(\alpha,\theta_2)  d\theta_2  -\epsilon
    <\sum_{h\in \{1,\cdots,H\}} \int_{\Theta_2^h} {v}^{x,y_h}(\theta_1,\theta_2)  d\theta_2
    <\sum_{h\in \{1,\cdots,H\}} \int_{\Theta_2^h} {v}^{x,y_h}(\alpha,\theta_2)  d\theta_2 +\epsilon, 
$
which proves the continuity in $\theta_1$. 
\end{proof}

Now, we are ready to prove our main result of equilibrium convergence in Theorem \ref{thm:main}. 
\begin{theorem}[Convergence of BNE by Discretization]
\label{thm:main}
A infinite Bayesian game $\Gamma$ has at least one BNE pair  $(F_0^{\mathcal{X}}\in \mathcal{F}^{\mathcal{X}},G_0^{\mathcal{Y}}\in \mathcal{G}^{\mathcal{Y}})$ in behavioral strategies. 
Moreover, there exists a sequence of discretized Bayesian games $\{\Gamma^{n_k}\}_{k\in \mathbb{Z}^+}$ such that $F_0^{x}(\theta_1)=\lim_{k\rightarrow \infty} F_{n_k}^{x}(\theta_1), \forall \theta_1\in \Theta, \forall x\in \mathcal{X}$ and $G_0^{y}(\theta_2)=\lim_{k\rightarrow \infty} G_{n_k}^{y}(\theta_2), \forall \theta_2\in \Theta, \forall y\in \mathcal{Y}$. 
\end{theorem}
\begin{proof}
We prove the theorem by contradiction. 
According to Theorem \ref{thm:first thm}, the sequence of mixed strategy pairs $(F^{ \mathcal{X}}_n\in \mathcal{F}^{\mathcal{X}}, G^{ \mathcal{Y}}_n\in \mathcal{G}^{\mathcal{Y}})$ will have a subsequence $(F^{ \mathcal{X}}_{n_k}, G^{ \mathcal{Y}}_{n_k})$ that converges weakly to a pair of strategies $(F^{ \mathcal{X}}_0\in \mathcal{F}^{\mathcal{X}}, G^{ \mathcal{Y}}_0\in \mathcal{G}^{\mathcal{Y}})$. 
Suppose the strategy pair $(F^{ \mathcal{X}}_0, G^{ \mathcal{Y}}_0)$ does not consist a BNE.  
Then, at least one of the two strategies is not a best response against the other. We may assume that $ G^{ \mathcal{Y}}_0$ is not optimal againt $F^{ \mathcal{X}}_0$. 
The second player's expected utility under the BNE of $\Gamma$ is
\begin{equation}
    w_0:=\sum_{x\in\mathcal{X}} \sum_{y\in\mathcal{Y}}     \int_0^1   \int_0^1  {v}^{x,y}(\theta_1,\theta_2)  dF_0^x(\theta_1) dG_0^y(\theta_2). 
\end{equation}
For each $n$, the second player's expected utility under the BNE of $\Gamma^n$ is
\begin{equation}
\begin{split}
        w_n & := \frac{1}{n^2} \sum_{x\in\mathcal{X}} \sum_{y\in\mathcal{Y}}  \sum_{i=1}^n \sum_{j=1}^n v^{x,y,n}_{i,j} s_i^{x,n} t_j^{y,n}  =\sum_{x\in\mathcal{X}} \sum_{y\in\mathcal{Y}}     \int_0^1   \int_0^1  {v}^{x,y}(\theta_1,\theta_2)  dF_n^x(\theta_1) dG_n^y(\theta_2). 
\end{split}
\end{equation}

Since $G_0^{\mathcal{Y}}$ is not an optimal response against $F_0^{\mathcal{X}}$ and Lemma \ref{lemma:purepolicy} shows that the deviation can be a pure strategy without loss of generality, there exists a set division of $\Theta_2$, i.e., $\Theta_2^h, \forall h\in \{1,2,...,H\} $, such that the deviation strategy $\bar{g}^{y_h}(\theta_2)=\mathbf{1}_{\{\theta_2\in \Theta_2^h \}}, \forall h\in \{1,2,...,H\}, \forall \theta_2\in \Theta_2$, achieves an expected utility larger than $w_0$. 
Then, there exists $\epsilon>0$ such that
\begin{equation*}
     \sum_{h\in \{1,\cdots,H\}} \int_{\Theta_2^h}    \sum_{x\in\mathcal{X}}     \int_0^1  {v}^{x,y_h}(\theta_1,\theta_2)  dF_0^x(\theta_1) d\theta_2 \geq w_0+4\epsilon. 
\end{equation*}

Based on the  continuity result in Lemma \ref{lemma:continuous} and the convergence result in Theorem \ref{thm:second thm}, for the set division $\{\Theta_2^h\}_{h\in \{1,\cdots,H\}}$, there exists $K_1$ such that if $k\geq K_1$, we have
\begin{equation*}
\begin{split}
        &\sum_{x\in\mathcal{X}} \int_0^1 \bigg[  \sum_{h\in \{1,\cdots,H\}} \int_{\Theta_2^h} {v}^{x,y_h}(\theta_1,\theta_2)  d\theta_2  \bigg]  dF_{n_k}^x(\theta_1)  \\
    > &  \sum_{x\in\mathcal{X}} \int_0^1 \bigg[  \sum_{h\in \{1,\cdots,H\}} \int_{\Theta_2^h} {v}^{x,y_h}(\theta_1,\theta_2)  d\theta_2  \bigg]  dF_0^x(\theta_1) -\epsilon,  
\end{split}
\end{equation*}
or equivalently, 
\begin{equation}
\label{eq:21}
    \frac{1}{n_k} \sum_{x\in\mathcal{X}}\sum_{i=1}^{n_k}  \bigg[  \sum_{h\in \{1,\cdots,H\}} \int_{\Theta_2^h} {v}^{x,y_h}( \frac{i}{n_k},\theta_2)  d\theta_2  \bigg] s_i^{x,n_k} > w_0+3\epsilon. 
\end{equation}
Theorem \ref{thm:second thm} also guarantees that there exists $K_2$ such that if $k\geq K_2$,
\begin{equation*}
   \begin{split}
        &\sum_{x\in\mathcal{X}}    \int_0^1  {v}^{x,y}(\theta_1,\theta_2)  dF_{n_k}^x(\theta_1) 
    <\sum_{x\in\mathcal{X}}  \int_0^1  {v}^{x,y}(\theta_1,\theta_2)  dF_0^x(\theta_1) 
    + \epsilon , \forall \theta_2\in  [0,1], \forall y\in \mathcal{Y},\\
    &\sum_{y\in\mathcal{Y}}    \int_0^1  {v}^{x,y}(\theta_1,\theta_2)  dG_{n_k}^x(\theta_2) 
    <    \sum_{y\in\mathcal{Y}}    \int_0^1  {v}^{x,y}(\theta_1,\theta_2)  dG_{0}^x(\theta_2) 
    + \epsilon , \forall \theta_1\in  [0,1], \forall x\in \mathcal{X}.
   \end{split}
\end{equation*}
Thus, we obtain $w_{n_k}<w_0+2\epsilon$ and from \eqref{eq:21}, we have
\begin{equation}
\label{eq:(28)}
        \frac{1}{n_k} \sum_{x\in\mathcal{X}}\sum_{i=1}^{n_k}  \bigg[  \sum_{h\in \{1,\cdots,H\}} \int_{\Theta_2^h} {v}^{x,y_h,n_k}( \frac{i}{n_k},\theta_2)  d\theta_2  \bigg] s_i^{x,n_k} > w_{n_k}+\epsilon.
\end{equation}

Owning to the continuity of  ${v}^{x,y_h}(\frac{i}{n_k},\theta_2)$ over $\theta_2$, $\int_{\Theta_2^h} {v}^{x,y_h,n_k}( \frac{i}{n_k},\theta_2)  d\theta_2$ is Riemann integrable. 
Since we discretize the entire type set $\Theta_2$ uniformly,  the length of the sub-interval of the partition is $\frac{1}{n_k} $. 
Thus, there exists $K_3$ such that if $k\geq K_3$, 
\begin{equation*}
     \frac{1}{n_k}  \sum_{h\in \{1,\cdots,H\}} \sum_{\frac{j}{n_k}\in \Theta_2^h} {v}^{x,y_h,n_k}( \frac{i}{n_k},\frac{j}{n_k})  d\theta_2  
     >  \sum_{h\in \{1,\cdots,H\}} \int_{\Theta_2^h} {v}^{x,y_h,n_k}( \frac{i}{n_k},\theta_2)  d\theta_2  -\epsilon, 
\end{equation*}
and so 
\begin{equation*}
     \begin{split}
         & \frac{1}{n_k}  \sum_{x\in\mathcal{X}}\sum_{i=1}^{n_k}  \bigg[ \sum_{h\in \{1,\cdots,H\}} \sum_{\frac{j}{n_k}\in \Theta_2^h} {v}^{x,y_h,n_k}( \frac{i}{n_k},\frac{j}{n_k})  \bigg] s_i^{x,n_k}\\
     > &   \sum_{x\in\mathcal{X}}\sum_{i=1}^{n_k}  \bigg[  \sum_{h\in \{1,\cdots,H\}} \int_{\Theta_2^h} {v}^{x,y_h,n_k}( \frac{i}{n_k},\theta_2)  d\theta_2 \bigg] s_i^{x,n_k} -n_k \cdot \epsilon. 
     \end{split}
\end{equation*}
Finally, combine with \eqref{eq:(28)}, we know that 
\begin{equation*}
     \frac{1}{(n_k)^2}   \sum_{x\in\mathcal{X}}\sum_{i=1}^{n_k}  \bigg[ \sum_{h\in \{1,\cdots,H\}} \sum_{\frac{j}{n_k}\in \Theta_2^h} {v}^{x,y_h,n_k}( \frac{i}{n_k},\frac{j}{n_k})  \bigg] s_i^{x,n_k} 
     > w_{n_k}, 
\end{equation*}
which leads to a contradiction as $t^{\mathcal{Y},n_k}$ was assumed to be an optimal response against $s^{\mathcal{X}, n_k}$ in the finite Bayesian game $\Gamma^{n_k}$. 
However, the second player achieves a higher expected utility under $s^{\mathcal{X}, n_k}$ if he adopts the pure BNE strategy $\bar{t}^{\mathcal{Y},n_k}$ whose $i$-th element $\bar{t}_i^{\mathcal{Y},n_k}$ satisfies $\bar{t}_i^{y_{h'},n_k}=\mathbf{1}_{\{h'=h\}}, \forall h'\in \{1,\cdots,H\}$,  if $\frac{i}{n_k}\in \Theta_2^h$. 
Therefore, the contradiction leads to the conclusion that $G^{\mathcal{Y}}_0$ is always optimal against $F^{\mathcal{X}}_0$ and the strategy pair ($F^{\mathcal{X}}_0$,$G^{\mathcal{Y}}_0$) consists a BNE in behavioral strategy for the infinite Bayesian game $\Gamma$. 
\end{proof}

\subsection{Algorithm to Compute $\varepsilon$-BNE of Infinite Bayesian Games}
Although Theorem \ref{thm:main} proves the asymptotic convergence of BNE, there is no finite-step performance guarantee. 
There exist counterexamples (see e.g., \cite{reeves2004computing}) where the finite approximation of an infinite game leads to misleading results. 
Due to the pathology, we construct Algorithm \ref{algorithm:espBNE} as follows to check whether a $\varepsilon$-BNE has been reached at some finite level $n$. 

\begin{algorithm}[h]
\SetAlgoLined
\textbf{Input} the infinite Bayesian game $\Gamma$, 
the approximation accuracy $\varepsilon>0$, and the maximum number of discretization $K$\;  
 Initialize the discretization level $n=1$\;
 \While{$n<K$}{
 Discretize $\Gamma$ via \eqref{eq:discretizaUtility} to obtain $\Gamma^n$\;
 Solve $\Gamma^n$ to obtain the equilibrium strategy pair $(s^{\mathcal{X},n},t^{\mathcal{Y},n})$\;
 Obtain the level-$n$ approximated strategy pair $(F_n^{\mathcal{X}},G_n^{\mathcal{Y}})$ for $\Gamma$ via \eqref{eq: FnGn}\;
  \uIf{   $(F_n^{\mathcal{X}},G_n^{\mathcal{Y}})$ consists a $\varepsilon$-BNE of $\Gamma$ in Definition \ref{def:epsBNE}   }{
\textbf{Terminate}\;
  }
   $n:=n+1$\;
}
  \textbf{Output}  the $\varepsilon$-BNE strategy $(F_n^{\mathcal{X}},G_n^{\mathcal{Y}})$ of the infinite Bayesian game $\Gamma$. 
 \caption{Compute $\varepsilon$-BNE of infinite Bayesian game $\Gamma$ \label{algorithm:espBNE}}
\end{algorithm}

To compute the BNE of finite Bayesian games in line $5$, we can construct the following  bilinear program ${C}^K$ (see Theorem 1 of \cite{huang2020dynamic}). Recall that the finite type set $\bar{\Theta}^n_i\subset \Theta$ contains the $n$ discrete types of player $i$. 
\begin{equation}
\begin{split}
 {[{C}^K]}:  & \max_{\sigma_{1},\sigma_{2},s_1,s_2}  \  
 \sum_{\theta_1\in \bar{\Theta}^n_1} \alpha_1(\theta_1)s_1(\theta_1)  +  \sum_{\theta_1\in \bar{\Theta}^n_1} 
 \alpha_1(\theta_1)\mathbb{E}_{\theta_{2}\sim b_1(\cdot|\theta_1), x\sim \sigma_1,y \sim \sigma_{2} } \allowbreak [ \bar{u}^{x,y}(\theta_1,\theta_{2})]
\\
&\quad +\sum_{\theta_2 \in \bar{\Theta}^n_2} \alpha_2(\theta_2)s_2(\theta_2)+ \sum_{\theta_2 \in \bar{\Theta}^n_2 } 
 \alpha_2(\theta_2)  \mathbb{E}_{\theta_{1}\sim b_2(\cdot|\theta_2), x\sim \sigma_1,y \sim \sigma_{2} } \allowbreak [ \bar{v}^{x,y}(\theta_1,\theta_{2})]
 \\
\text {s.t.}  \quad \quad   &(a)  \quad  
 \mathbb{E}_{\theta_{1}\sim b_2(\cdot|\theta_2), x\sim \sigma_1 } \allowbreak [ \bar{v}^{x,y}(\theta_1,\theta_{2})]  \leq -s_2(\theta_2) , \forall \theta_2\in \bar{\Theta}^n_2, \forall y\in \mathcal{Y}, \\
 & (b) \quad 
\sum_{x\in \mathcal{X}} \sigma_1(x|\theta_1)=1, \sigma_1(x|\theta_1)\geq 0,  \forall \theta_1\in \bar{\Theta}^n_1, \\
&(c)  \quad  
\mathbb{E}_{\theta_{2}\sim b_1(\cdot|\theta_1),y \sim \sigma_{2} } \allowbreak [ \bar{u}^{x,y}(\theta_1,\theta_{2})] \leq -s_1(\theta_1), \ \forall \theta_1\in \bar{\Theta}^n_1, \forall x\in \mathcal{X}, \\
& (d) \quad 
\sum_{y\in \mathcal{Y}} \sigma_2(y|x,\theta_2)=1, \sigma_2(y|\theta_2)\geq 0,  \forall \theta_2\in \bar{\Theta}^n_2.
\label{eq: constrained optimization}
\end{split}
\end{equation}
Note that  $\alpha_1(\theta_1), \forall \theta_1\in \bar{\Theta}^n_1$ and $\alpha_2(\theta_2), \forall \theta_2\in \bar{\Theta}^n_2$, are not decision variables and can be any strictly positive and finite numbers. 
Thus, we have the freedom to pick them properly to obtain a linear program rather than a bilinear program under certain conditions as shown in Proposition \ref{prop:LP}. 
\begin{proposition}[Linear Program Reformulation]
\label{prop:LP}
If there exists $m_i(\theta_i)>0, \forall i\in \{1,2\}, \forall \theta_i\in \bar{\Theta}^n_i$, such that
$m_2(\theta_2) \bar{u}^{x,y}(\theta_1,\theta_{2})=-m_1(\theta_1) \bar{v}^{x,y}(\theta_1,\theta_{2})$ holds for all $ x\in \mathcal{X},y\in \mathcal{Y},\theta_1\in \bar{\Theta}^n_1,\theta_2\in \bar{\Theta}^n_2$, then we can pick $\alpha_i(\theta_i)=\bar{b}_i(\theta_i)/ m_i(\theta_i)>0$ to make 
$C^K$ a linear program. 
\end{proposition}
\begin{proof}
It is straightforward to verify that two bilinear terms always sum up to $0$, i.e., 
\begin{equation*}
    \sum_{\theta_1\in \bar{\Theta}^n_1} 
 \alpha_1(\theta_1)\mathbb{E}_{\theta_{2}\sim b_1(\cdot|\theta_1), x\sim \sigma_1,y \sim \sigma_{2} } \allowbreak [ \bar{u}^{x,y}(\theta_1,\theta_{2})]+\sum_{\theta_2 \in \bar{\Theta}^n_2 } 
 \alpha_2(\theta_2)  \mathbb{E}_{\theta_{1}\sim b_2(\cdot|\theta_2), x\sim \sigma_1,y \sim \sigma_{2} } \allowbreak [ \bar{v}^{x,y}(\theta_1,\theta_{2})] \equiv 0,  
\end{equation*}
for all feasible strategy pair $\sigma_1,\sigma_2$, if we choose  $\alpha_i(\theta_i)=\bar{b}_i(\theta_i)/ m_i(\theta_i)$. 
\end{proof}
Note that the condition $m_i(\theta_i)=1, \forall i\in \{1,2\}, \forall \theta_i\in \bar{\Theta}^n_i$, results in a zero-sum finite Bayesian game. Then, we can recast $C^K$ as a linear program by picking $\alpha_i(\theta_i)=\bar{b}_i(\theta_i), \forall \theta_i\in \bar{\Theta}^n_1$, which coincides with the existing result in  \cite{ponssard1980lp}. 

\bibliographystyle{plain}
\bibliography{reference}

\end{document}